\documentclass[11pt]{article}
\usepackage[margin=1in]{geometry}
\usepackage{amsmath,amssymb,amsthm}
\usepackage{tikz}
\usetikzlibrary{arrows.meta,positioning,calc}
\usepackage{float}
\usepackage{mathrsfs}
\usepackage[hidelinks]{hyperref}
\hypersetup{
  pdftitle={At Most Two Infinite Blue Clusters in the CMR Representation of the Edwards--Anderson Spin Glass},
  pdfauthor={Yan Ru Pei}
}

\newtheorem{theorem}{Theorem}
\newtheorem{lemma}[theorem]{Lemma}
\newtheorem{proposition}[theorem]{Proposition}
\newtheorem{corollary}[theorem]{Corollary}
\newtheorem{definition}[theorem]{Definition}
\numberwithin{theorem}{section}
\theoremstyle{remark}
\newtheorem{remark}[theorem]{Remark}

\title{At Most Two Infinite Blue Clusters in the CMR Representation of the Edwards--Anderson Spin Glass}
\author{Yan Ru Pei}
\date{}

\begin{document}
\maketitle

\begin{abstract}
The two-replica Chayes--Machta--Redner (CMR) representation is one of the main proposed geometric signatures of spin-glass order in the short-range Edwards--Anderson model. Mean-field arguments and recent numerics suggest that the low-temperature phase should exhibit two macroscopic blue clusters carrying opposite overlap signs. We prove a rigorous structural constraint in this direction. For any subsequential local weak limit of the standard periodic-torus joint laws on disorder, two spin replicas, and CMR bond variables, the blue subgraph contains at most two infinite connected components; if two exist, then they lie in a common infinite grey cluster and belong to opposite overlap-parity classes. The main obstacle is that the labelled blue geometry does not permit unrestricted insertions across overlap classes, and no positive-association input is available, so the usual Burton--Keane and random-cluster arguments do not apply directly. We isolate an abstract multicolour Burton--Keane proposition based on finite-box label-class coalescence and verify its hypothesis for CMR blue bonds by resampling the full joint measure. As auxiliary input, we establish finite energy and a percolation transition for the grey subgraph via local resampling of the disorder and a parity-based Peierls estimate. These results do not prove the existence of infinite blue clusters or a spin-glass phase transition, but they give a rigorous upper bound compatible with the two-cluster picture for short-range spin glasses.
\end{abstract}

\section{Introduction and main result}\label{sec:intro}

\subsection{Motivation}

Spin glasses are disordered magnetic systems in which competing interactions can prevent the emergence of a single preferred ordered state.
A central problem is to understand whether the low-temperature phase has genuine long-range order, and if so how that order should be detected.
Unlike ferromagnets, where order is visible through magnetization and through Fortuin--Kasteleyn percolation, spin glasses have no fixed direction of magnetization.
Their standard order parameter is instead the overlap between two independently sampled equilibrium configurations.

This makes graphical representations especially useful: they turn questions about overlap into questions about random geometry.
The two-replica Chayes--Machta--Redner (CMR) representation~\cite{CMR1998} associates to a pair of spin configurations a bond process with blue bonds, red bonds, and closed bonds.
Blue clusters are the natural geometric objects associated with regions where both replicas simultaneously satisfy the disorder.
Machta, Newman, and Stein~\cite{MNS2007} argued that the low-temperature phase should be reflected in the emergence of two macroscopic blue clusters of opposite overlap sign.
Recent simulations in two and three dimensions by M\"unster and Weigel~\cite{MW2023,MW2025} support this picture.

For short-range Edwards--Anderson models, even the structural part of this picture is not automatic.
Frustration imposes nonlocal cycle constraints, while the labelled blue geometry forbids unrestricted insertions across overlap classes and provides no positive-association input.
Thus, before one can identify the spin-glass order parameter with blue-cluster densities, one must first answer a more basic geometric question: how many infinite blue clusters can coexist?

This note proves that the answer is at most two.
More precisely, for any joint measure obtained from the periodic-torus construction specified below, there is at most one infinite blue cluster in each overlap-parity class, and hence at most two infinite blue clusters in total.
If two exist, they lie in the same infinite grey cluster and belong to opposite path-parity classes.
The proof works in the full joint measure over disorder, spins, and bonds. As auxiliary input, we establish finite energy and a percolation transition for the grey subgraph (blue$\,\cup\,$red) by combining local resampling of the couplings with a parity-based Peierls argument.

We work with the Edwards--Anderson (EA) Ising spin glass on the hypercubic lattice $\mathbb{Z}^d$ ($d \geq 2$), with vertex set $V$ and nearest-neighbor edge set $E$. The disorder is a collection of random couplings $J = (J_e)_{e \in E}$. For a finite $\Lambda\subset V$, boundary condition $\xi\in\{\pm1\}^{\Lambda^c}$, and inverse temperature $\beta>0$, the finite-volume EA kernel is
\begin{equation}\label{eq:gibbs}
  \mu_{J,\beta,\Lambda}^{\xi}(\sigma_\Lambda)
  \;\propto\;
  \exp\Biggl(\beta
    \sum_{\substack{e=\{x,y\}\in E\\ e\cap\Lambda\neq\emptyset}}
      J_e\,\sigma_x^\xi\sigma_y^\xi\Biggr),
\end{equation}
where $\sigma^\xi$ agrees with $\sigma_\Lambda$ on $\Lambda$ and with $\xi$ on $\Lambda^c$.
An infinite-volume EA Gibbs measure is a probability law whose finite-box conditional distributions are these kernels.
Positive couplings favor alignment, negative couplings favor anti-alignment, and loops carrying an odd number of negative couplings are frustrated.

\subsection{The CMR Edwards--Sokal coupling on the Ising spin glass}

The disorder distribution $\mathscr{P}$ (the common law of a single coupling $J_e$) determines both the frustration structure and the analytical tools available.

\begin{definition}[Disorder conditions]\label{def:disorder}
We assume throughout:
\begin{enumerate}
\item[\textup{(D1)}] \textbf{Independence.} The couplings $(J_e)_{e \in E}$ are i.i.d.
\item[\textup{(D2)}] \textbf{Two-sided.} $\mathscr{P}(J > 0) > 0$, $\;\mathscr{P}(J < 0) > 0$, and $\;\mathscr{P}(\{0\}) = 0$.
\end{enumerate}
\end{definition}

\begin{remark}
The \textbf{bimodal} ($\pm J_0$) and \textbf{Gaussian} ($\mathscr{P} = \mathcal{N}(0, \Delta^2)$) distributions both satisfy~(D1)--(D2).
The two-sidedness condition ensures that frustration occurs with positive density and that the local resampling of couplings (Section~\ref{sec:bayesian}) can access either coupling sign.
One-sided distributions ($\mathscr{P}(J > 0) = 1$, the ferromagnet) are excluded by~(D2); their two-replica CMR geometry is not treated here.
\end{remark}

\begin{definition}[CMR Edwards--Sokal coupling]\label{def:CMR}
The CMR representation~\cite{CMR1998} is an Edwards--Sokal coupling for two replicas $(\sigma, \tau)$ sampled independently conditional on the same disorder $J$ in finite volume.
A bond $e = \{x,y\}$ is \emph{satisfied} by replica~$\sigma$ if $J_e \sigma_x \sigma_y > 0$.
Given $J$ and $\beta$, set $r_e := e^{-2\beta|J_e|}$.
On a finite graph $G=(V_G,E_G)$, after the spin-independent normalization that makes a satisfied two-replica edge have weight~$1$, the joint spin--bond weight is
% Important: do not replace this local table by a single Bernoulli
% factor for all non-closed bonds. Red and blue have different normalized
% weights: blue 1-r^2, red r(1-r), closed r^2. The column sums must recover
% the two-replica edge weights 1,r,r,r^2.
\begin{equation}\label{eq:ES}
  X_G(\sigma, \tau, \eta; J) = \prod_{e \in E_G} w_e(\sigma, \tau, \eta_e; J_e),
\end{equation}
where the local weights are

\begin{center}
\renewcommand{\arraystretch}{1.1}
\begin{tabular}{lcccc}
& \textbf{Both satisfy} & \textbf{$\sigma$ only} & \textbf{$\tau$ only} & \textbf{Neither satisfies} \\
& $(+,+)$ & $(+,-)$ & $(-,+)$ & $(-,-)$ \\
\hline
\textup{blue} & $1-r_e^2$ & $0$ & $0$ & $0$ \\
\textup{red} & $0$ & $r_e(1-r_e)$ & $r_e(1-r_e)$ & $0$ \\
\textup{closed} & $r_e^2$ & $r_e^2$ & $r_e^2$ & $r_e^2$ \\
\end{tabular}
\end{center}

Thus blue bonds are allowed only when both replicas satisfy $e$, red bonds only when exactly one replica satisfies $e$, and closed bonds are always allowed.
The column sums are $1,r_e,r_e,r_e^2$, which recover the normalized two-replica Edwards--Anderson edge weights after summing over $\eta_e$.
Consequently, conditional on a fixed $J$ in every finite volume, summing over all bond configurations gives the product of two independent Edwards--Anderson Gibbs measures as the spin marginal: the CMR variables augment rather than alter the two-replica EA model.
In infinite volume, normalizing the applicable column of the same table specifies the local CMR bond kernel conditional on $(J,\sigma,\tau)$; it is not an infinite-product density.

\end{definition}

\begin{remark}[Marginal measure]\label{rem:weights}
In finite volume, the \textbf{CMR bond marginal} is obtained by integrating out the spins $(\sigma, \tau)$. Its unnormalized weights are determined by bond fugacities and \textbf{geometric factors} from the cluster weight formula $V(\eta) = F(\eta)\,D(\eta)\,2^{K(\eta)}\,2^{C_b(\eta)}$~\cite[eq.~(9)]{CMR1998theory}, where $K(\eta)$ is the number of grey clusters and $C_b(\eta)$ is the number of blue clusters. The merge-case weights needed later are recorded in Remark~\ref{rem:merge-weights}.
\end{remark}

\subsection{Bond subgraphs and overlap structure}

Define the \textbf{local overlap} $q_i := \sigma_i \tau_i \in \{\pm 1\}$.

\begin{definition}[Bond subgraphs]\label{def:subgraphs}
Given a CMR bond configuration $\eta \in \{\text{blue}, \text{red}, \text{closed}\}^E$:
\begin{itemize}
\item The \textbf{blue-bond subgraph} $G_{\mathrm{blue}} = (V, E_{\mathrm{blue}})$ is the spanning subgraph of $\mathbb{Z}^d$ with edge set $E_{\mathrm{blue}} := \{e \in E : \eta_e = \text{blue}\}$.
\item The \textbf{grey-bond subgraph} $G_{\mathrm{grey}} = (V, E_{\mathrm{grey}})$ is the spanning subgraph with edge set $E_{\mathrm{grey}} := \{e \in E : \eta_e \in \{\text{blue}, \text{red}\}\}$.
\end{itemize}
Connected components of $G_{\mathrm{blue}}$ are called \textbf{blue clusters}; connected components of $G_{\mathrm{grey}}$ are called \textbf{grey clusters}.
Each grey cluster decomposes into at most two nonempty overlap parity classes (Definition~\ref{def:parity}), and each blue cluster is contained in a single parity class.
\end{definition}

\begin{remark}[Overlap alignment from the Edwards--Sokal coupling]
The allowance rules in~\eqref{eq:ES} enforce rigid overlap constraints within the bond subgraphs: along a \textbf{blue bond} (double satisfaction), both replicas satisfy the bond, so $\sigma_x \sigma_y = \tau_x \tau_y = \mathrm{sign}(J_e)$, which forces $q_x = \sigma_x \tau_x = \sigma_y \tau_y = q_y$ and hence makes the overlap constant on each blue cluster; along a \textbf{red bond} (single satisfaction), exactly one replica satisfies the bond, so $\sigma_x \sigma_y \neq \tau_x \tau_y$, which forces $q_x = -q_y$ and hence makes the overlap flip across every red bond.
\end{remark}

The overlap alignment imposes two global constraints on cycles in the bond configuration:
\begin{enumerate}
\item[(C1)] Every cycle contains an \textbf{even} number of red bonds.
\item[(C2)] No frustrated cycle (one with $\prod_{e \in \gamma} \mathrm{sign}(J_e) = -1$) consists entirely of blue bonds.
\end{enumerate}

\begin{definition}[Overlap parity class]\label{def:parity}
Under the joint measure on $(J, \sigma, \tau, \eta)$, each grey cluster decomposes into at most two nonempty \textbf{overlap parity classes}: the set of vertices with overlap $q_i = +1$ and the set with $q_i = -1$, either of which may be empty.
Every blue cluster is contained in a single parity class (since $q$ is constant on blue clusters), and sites across a red bond have opposite overlap.
\end{definition}

\begin{remark}[Parity classes in the bond marginal]\label{rem:parity-marginal}
The overlap parity class is also well-defined from the bond configuration alone: two vertices in the same grey cluster are in the \textbf{same} parity class if every grey path between them crosses an \textbf{even} number of red bonds, and in \textbf{opposite} classes if every such path crosses an odd number.
By~(C1), the parity of red bonds along any path is path-independent, so this is consistent and partitions each grey cluster into at most two nonempty classes---which, under the spin lift, coincide with the nonempty $q = +1$ and $q = -1$ sets of Definition~\ref{def:parity}.
Thus we use \emph{label class} for the abstract labelled graph, \emph{overlap-sign class} for the sets $q=\pm1$ in the joint CMR representation, and \emph{path-parity class} for their bond-measurable description inside a grey cluster.
\end{remark}

\subsection{Main result}

The main result of this note is the following.
Let $\mathcal{J} := \operatorname{supp}(\mathscr{P})$ denote the single-edge coupling space.
On the torus $\mathbb{T}_L^d$, let
\begin{equation}\label{eq:torus-joint-law}
  \mathbf{P}_{\beta,L}^{\mathrm{per}}(\mathrm{d}J,\mathrm{d}\sigma,\mathrm{d}\tau,\mathrm{d}\eta)
  := \nu_L(\mathrm{d}J)\,
     \mu_{J,\beta,L}^{\mathrm{per}}(\mathrm{d}\sigma)\,
     \mu_{J,\beta,L}^{\mathrm{per}}(\mathrm{d}\tau)\,
     K_{J,\sigma,\tau,L}(\mathrm{d}\eta),
\end{equation}
where $\nu_L=\mathscr{P}^{\otimes E(\mathbb{T}_L^d)}$ and $K_{J,\sigma,\tau,L}$ is the CMR bond kernel from Definition~\ref{def:CMR}.
We call any subsequential local weak limit of these laws, with the two replicas and bonds included before taking $L\to\infty$, an \emph{admissible joint measure}.
For $\mathscr{P}^{\otimes E}$-almost every $J$, the conditional two-replica spin law is Gibbs for the noninteracting two-replica EA specification, whose finite-box DLR kernel is the product of the two one-replica kernels.
This local product specification does not imply that the infinite-volume two-replica law is itself a product measure; the replicas may remain correlated through the Gibbs state selected in the limit.
Each individual replica marginal is nevertheless an EA Gibbs measure; see~\cite[Proposition~4.12, Remark~4.13, and Lemma~B.4]{Newman1997}.

\begin{theorem}\label{thm:main}
Let $d \geq 2$. For each $\beta > 0$, let $\mathbf{P}_\beta$ be an admissible joint measure on
\[
  \Omega = \mathcal{J}^E \times \{\pm 1\}^V \times \{\pm 1\}^V \times \{\text{blue, red, closed}\}^E
\]
as constructed above.
\begin{enumerate}
\item[\textup{(a)}] \textbf{Grey percolation transition.}
  There exist $0 < \beta_{\mathrm{sub}} \leq \beta_{\mathrm{grey}} < \infty$ (depending on $d$ and $\mathscr{P}$) such that:
  for $\beta < \beta_{\mathrm{sub}}$, $G_{\mathrm{grey}}$ contains no infinite grey cluster;
  for $\beta > \beta_{\mathrm{grey}}$, $G_{\mathrm{grey}}$ contains exactly one infinite grey cluster;
  both statements holding $\mathbf{P}_\beta$-a.s.
\item[\textup{(b)}] \textbf{Blue-cluster bound.}
  For every $\beta > 0$, $\mathbf{P}_\beta$-a.s., $G_{\mathrm{blue}}$ contains at most two infinite blue clusters.
  If two distinct infinite blue clusters exist, then they are contained in the same grey cluster, and every grey path between them crosses an odd number of red bonds.
\end{enumerate}
\end{theorem}

\begin{remark}
Part~(a) establishes a grey-bond percolation transition but does not imply the existence of infinite blue clusters or a spin glass phase transition---the infinite grey cluster could be sustained entirely by red bonds.
Part~(b) is a structural constraint that holds at all temperatures and is purely a statement about the bond configuration: inside each grey cluster, the path-parity classes are determined by~(C1) without reference to the spins.
If two infinite blue clusters exist, Corollary~\ref{cor:grey-unique} forces them into the same infinite grey cluster, where the two path-parity classes correspond under the spin lift to overlap signs $q = +1$ and $q = -1$.
Thus part~(b) can equivalently be read as ``at most one infinite blue cluster per overlap sign inside the infinite grey cluster.''
\end{remark}

\paragraph{Ergodic reduction.}
Fix $\beta > 0$.
For the proof of part~\textup{(b)}, we may pass to an ergodic component at this fixed temperature, since the event in part~\textup{(b)} is translation-invariant and the conclusion is preserved under convex mixtures.
The translation-invariant sigma-field is tail-measurable modulo $\mathbf{P}_\beta$; see~\cite[Lemma~6.67]{FV2017}.
It is consequently contained, modulo null sets, in every finite-volume exterior sigma-field.
By the same conditional-expectation argument used for tail-conditioned Gibbs measures in~\cite[Proposition~6.69]{FV2017}, disintegrating over it therefore leaves the local spin DLR and CMR bond kernels unchanged.
Likewise, for every finite edge set $S$, the full tail sigma-field is contained in the completion of $\mathcal{T}_S'$ from Proposition~\ref{prop:bayes}, so the conditional law of $J_S$ given $\mathcal{T}_S'$ and hence its equivalence to $\mathscr{P}^{\otimes S}$ are unchanged in almost every ergodic component.
Finally, each component retains disorder marginal $\mathscr{P}^{\otimes E}$ because the i.i.d.\ disorder law is itself translation-ergodic.
Thus in Section~\ref{sec:bk} and its consequences we write $\mathbf{P}_\beta$ for an ergodic component; the general case follows by integrating over components.
\paragraph{Notation for measures.}
Throughout, $\mathscr{P}$ denotes the common distribution of a single coupling $J_e$, and $\nu = \bigotimes_e \mathscr{P}$ is the i.i.d.\ product measure on the full disorder $J = (J_e)_{e \in E}$.
The \textbf{joint measure} $\mathbf{P}_\beta$ on $(J, \sigma, \tau, \eta)$ is defined above; its conditional $\mathbf{P}_{\beta,J} := \mathbf{P}_\beta(\,\cdot \mid J)$ is the quenched measure on $(\sigma, \tau, \eta)$ at fixed disorder.
For the selected admissible law, write
\[
  \rho_{\beta,J} := \mathbf{P}_{\beta,J}\circ(\sigma,\tau)^{-1},
  \qquad
  \zeta_{\beta,J} := \mathbf{P}_{\beta,J}\circ\eta^{-1}
\]
for its two-replica spin marginal and CMR bond marginal, respectively.
This notation records that these marginals belong to the selected admissible limit; neither a unique infinite-volume EA state nor a unique CMR bond marginal is assumed.

\paragraph{Proof strategy.}
The proof has two largely separate tracks; the dependency diagram below summarizes the main inputs used in each.

\smallskip\noindent
\textit{Grey percolation transition (part~(a), Section~\ref{sec:energy}).}
The grey (blue$\,\cup\,$red) bond process has full finite energy: a case analysis of the conditional bond probabilities shows that every bond can be both inserted and deleted with positive probability under the joint measure $\mathbf{P}_\beta$.
The local disorder-resampling input is needed only in the same-overlap endpoint case, where one sign of the coupling can force the bond closed; the analogous single-edge partition-function comparison in a frustrated FK model appears in De~Santis and Gandolfi~\cite[Lemma~3.1]{DSG1999}.
Uniqueness of the infinite grey cluster then follows from the standard Burton--Keane argument~\cite{BK1989}.
Existence for large $\beta$ uses the Peierls contour argument of De~Santis and Gandolfi~\cite{DSG1999}, adapted to the two-replica setting: we replace their frustration-based $\mathbb{Z}_2$ partition of the contour with a red-bond parity partition (via the balance theorem for signed graphs), which makes the frustration constraint~(C2) irrelevant and yields a coupling-by-coupling contour estimate. This gives quenched existence when the couplings are bounded away from zero and annealed existence for general laws satisfying~\textup{(D1)--(D2)}.
Absence for small $\beta$ follows from a truncation/direct-domination argument: strong couplings are rare by~\textup{(D1)}, while weak edges have uniformly small grey-open probability.

\smallskip\noindent
\textit{Blue-cluster uniqueness per parity class (part~(b), Section~\ref{sec:bk}).}
We first prove an abstract multicolour Burton--Keane proposition for compatible labelled subgraphs satisfying \textbf{label-class coalescence}: for each fixed label, all incident infinite exterior components can be joined through a finite box with positive conditional probability.
This is a label-restricted analogue of the merge tolerance of Halberstam and Hutchcroft~\cite[Lemma~3.4]{HH2023}, weaker in what it can merge because it says nothing about connectivity classes carrying different labels.
For CMR blue bonds the labels are the overlap values $q=\pm1$.
The labelled blue subgraph does not admit unrestricted insertions across these classes, but a stronger model-specific box-filling event can be produced by resampling $(J,\sigma,\tau,\eta)$ inside the box.
The end-counting obstruction is the mass-transport theorem used by Halberstam and Hutchcroft~\cite[Proposition~3.3]{HH2023}: every component of a translation-invariant random subgraph of $\mathbb{Z}^d$ has at most two ends almost surely.
Applying the same argument separately for $q=+1$ and $q=-1$ yields at most one infinite blue cluster per parity class, hence at most two in total.

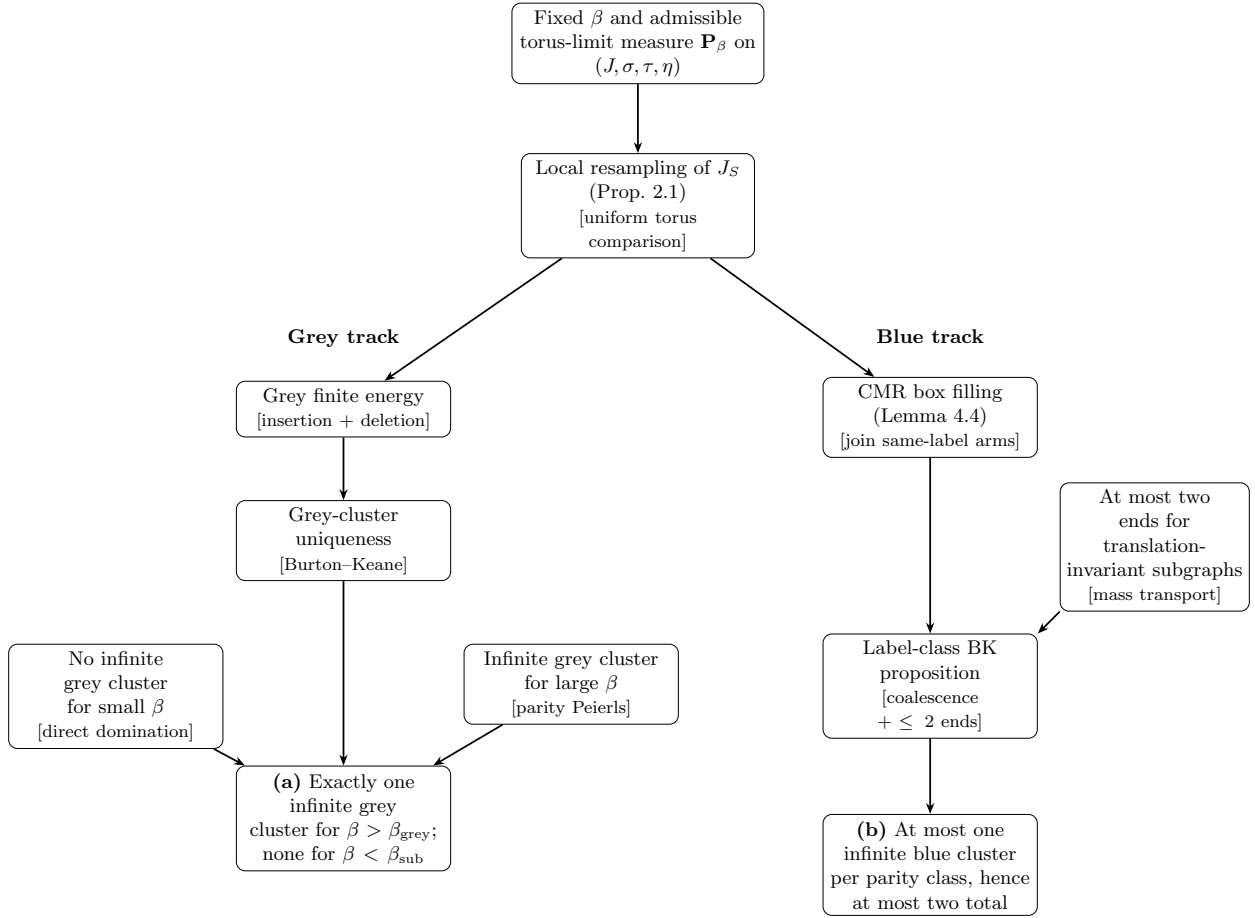
\begin{figure}[H]
\centering
\resizebox{\linewidth}{!}{%
\begin{tikzpicture}[
    node distance=1.0cm and 0.9cm,
    box/.style={draw, rounded corners, text width=3.2cm, align=center,
                font=\footnotesize, minimum height=0.8cm, inner sep=3pt},
    track/.style={font=\footnotesize\bfseries},
    arr/.style={-{Stealth[length=5pt]}, thick},
  ]

  \node[box, text width=3.8cm] (joint) {Fixed $\beta$ and admissible\\torus-limit measure $\mathbf{P}_\beta$ on\\$(J, \sigma, \tau, \eta)$};

  \node[box, below=1.1cm of joint, text width=3.5cm] (bayes) {Local resampling of $J_S$\\(Prop.~\ref{prop:bayes})\\{\scriptsize [uniform torus comparison]}};

  \node[track, below left=1.0cm and 1.8cm of bayes] (greytitle) {Grey track};
  \node[track, below right=1.0cm and 1.8cm of bayes] (bluetitle) {Blue track};

  \node[box, below=0.4cm of greytitle] (greyfe) {Grey finite energy\\{\scriptsize [insertion + deletion]}};
  \node[box, below=1.0cm of greyfe] (greyuniq) {Grey-cluster uniqueness\\{\scriptsize [Burton--Keane]}};
  \node[box, below left=1.0cm and 0.2cm of greyuniq] (greysub) {No infinite grey cluster\\for small $\beta$\\{\scriptsize [direct domination]}};
  \node[box, below right=1.0cm and 0.2cm of greyuniq] (greysup) {Infinite grey cluster\\for large $\beta$\\{\scriptsize [parity Peierls]}};
  \node[box, below=1.2cm of $(greysub)!0.5!(greysup)$] (greycombined)
    {\textbf{(a)} Exactly one infinite grey\\cluster for $\beta > \beta_{\mathrm{grey}}$;\\none for $\beta < \beta_{\mathrm{sub}}$};

  \node[box, below=0.4cm of bluetitle] (boxresample) {CMR box filling\\(Lemma~\ref{lem:resample})\\{\scriptsize [join same-label arms]}};
  \node[box, below=2.8cm of boxresample] (bluebk) {Label-class BK proposition\\{\scriptsize [coalescence + $\leq 2$ ends]}};
  \node[box, text width=2.8cm, above right=0.35cm and 0.35cm of bluebk] (twoends)
    {At most two ends for\\translation-invariant subgraphs\\{\scriptsize [mass transport]}};
  \node[box, below=1.2cm of bluebk] (bluecombined)
    {\textbf{(b)} At most one infinite blue cluster\\per parity class, hence at most two total};

  \draw[arr] (joint) -- (bayes);
  \draw[arr] (bayes) -- (greyfe);
  \draw[arr] (bayes) -- (boxresample);
  \draw[arr] (greyfe) -- (greyuniq);
  \draw[arr] (greyuniq) -- (greycombined);
  \draw[arr] (greysub) -- (greycombined);
  \draw[arr] (greysup) -- (greycombined);
  \draw[arr] (boxresample) -- (bluebk);
  \draw[arr] (twoends.south west) -- (bluebk.north east);
  \draw[arr] (bluebk) -- (bluecombined);

\end{tikzpicture}
}
\end{figure}

%% Local disorder-resampling section
%% \input from blue-bond-uniqueness.tex between intro and grey-cluster

\section{Local resampling of couplings}\label{sec:bayesian}

In the ferromagnetic FK/Edwards--Sokal representation ($J_e > 0$ for all $e$), the bond marginal at fixed $J$ has full finite energy: every bond can be inserted or deleted with positive conditional probability.
For spin glasses, this fails at fixed $J$---the frustration constraint~(C2) can force a bond closed with probability one---but it is recovered in the \textbf{joint measure} $\mathbf{P}_\beta$ on $(J, \sigma, \tau, \eta)$ by treating the couplings as variables to be resampled.
Conditioning on the spins tilts the product disorder law by a local Radon--Nikodym factor; in finite volume, its normalization can be expressed as a cavity partition-function ratio.

This section isolates the only disorder-resampling input needed later: after freezing the spins everywhere and the bonds outside a finite edge set, the conditional law of the couplings on that edge set still has strictly positive density with respect to the i.i.d.\ prior.
The result (Proposition~\ref{prop:bayes}) is invoked in the grey finite-energy argument and in the CMR box-filling lemma that verifies label-class coalescence.

\subsection{Conditional law of finitely many couplings}\label{sec:fv-conditional}

Let $S \subset E$ be finite, and let
\[
  \mathcal{T}_S' := \sigma\bigl(J_{S^c}, \sigma, \tau, \eta_{S^c}\bigr),
  \qquad S^c := E \setminus S,
\]
where $\eta_{S^c} = (\eta_e)_{e \notin S}$.

\begin{proposition}[Local equivalence of coupling conditionals]\label{prop:bayes}
Let $S \subset E$ be finite.
Under $\mathbf{P}_\beta$, the regular conditional law of $J_S$ given $\mathcal{T}_S'$ is absolutely continuous with respect to $\bigotimes_{e \in S} \mathscr{P}$, and its density is $\mathbf{P}_\beta$-a.s.\ strictly positive.
Consequently, for every measurable $A \subset \mathrm{supp}(\mathscr{P})^{|S|}$ with $\bigotimes_{e \in S} \mathscr{P}(A) > 0$,
\begin{equation}\label{eq:bayes-positive}
  \mathbf{P}_\beta\bigl(J_S \in A \mid \mathcal{T}_S'\bigr)
  \;>\; 0.
\end{equation}
In particular, by~\textup{(D2)} couplings of either sign can be sampled with positive probability.
\end{proposition}

\begin{proof}
On a finite torus, change the couplings on $S$ from $j$ to $j'$ and set
$D(j,j'):=\sum_{e\in S}|j_e-j_e'|$.
For one replica, the Boltzmann factor and the ratio of partition functions are each between $e^{-\beta D(j,j')}$ and $e^{\beta D(j,j')}$; hence for two replicas
\begin{equation}\label{eq:torus-rn-comparison}
  e^{-4\beta D(j,j')}
  \leq
  \frac{\mathrm{d}(\mu_{J_{S^c}\cup j',\beta,L}^{\mathrm{per}})^{\otimes 2}}
       {\mathrm{d}(\mu_{J_{S^c}\cup j,\beta,L}^{\mathrm{per}})^{\otimes 2}}
       (\sigma,\tau)
  \leq
  e^{4\beta D(j,j')}.
\end{equation}
After the spins and $\eta_{S^c}$ are fixed, the CMR factors outside $S$ are unchanged and the bond variables on $S$ sum to one.
Thus~\eqref{eq:torus-rn-comparison} is also the likelihood-ratio bound for the conditional density of $J_S$ relative to $\mathscr{P}^{\otimes S}$.
The bound is uniform in $L$ on bounded coupling sets.
Integrating it against cylinder events in $\mathcal{T}_S'$, passing to the defining local weak limit, and then using martingale convergence preserves the comparison under conditioning on all of $\mathcal{T}_S'$; for unbounded $\mathscr{P}$, first restrict $J_S$ to bounded sets and then exhaust its support.
The limiting conditional density is therefore finite and strictly positive $\mathscr{P}^{\otimes S}$-a.e., which gives~\eqref{eq:bayes-positive}.
\end{proof}

\begin{remark}
The CMR verification of label-class coalescence uses three local properties of the torus construction: the finite-coupling equivalence above, finite-box non-nullness of the two-replica spin specification, and the CMR bond kernel of Definition~\ref{def:CMR}.
Consequently that argument also applies to any translation-invariant joint CMR law for which these properties are assumed directly, in particular to coupling-covariant extensions of the construction.
\end{remark}

\section{Grey-cluster analysis}\label{sec:energy}

\subsection{Finite energy of the grey cluster}

One might expect the grey-bond process to have finite energy even in the quenched setting, since the frustration constraint~(C2) applies only to all-blue cycles.
However, a case analysis reveals that local disorder resampling is still needed in one case.

\begin{proposition}[Finite energy of the grey cluster]\label{prop:grey-fe}
Let $\mathcal{T}_b'$ denote the $\sigma$-algebra generated by all variables except $(\eta_b, J_b)$.
The joint measure $\mathbf{P}_\beta$ satisfies: for every bond $b \in E$,
\[
  0 < \mathbf{P}_\beta(\eta_b \in \{\textup{blue}, \textup{red}\} \mid \mathcal{T}_b') < 1 \quad \mathbf{P}_\beta\text{-a.s.}
\]
When the endpoints have opposite overlap this holds even conditional on $J_b$ (i.e., with $\mathcal{T}_b' $ replaced by $\mathcal{T}_b := \sigma(\mathcal{T}_b' , J_b)$).
When the endpoints have the same overlap, one sign of $J_b$ can force the bond closed, which is why the quenched measure at fixed $J$ fails to have finite energy.
\end{proposition}

\begin{proof}
Let $b = \{x, y\}$.
The conditioning $\mathcal{T}_b'$ freezes the spins, so the conditional law of $\eta_b$ is read directly from the local Edwards--Sokal table in Definition~\ref{def:CMR} after resampling $J_b$.
% Here \mathcal{T}'_b freezes the spin lift. Therefore the single-edge
% finite-energy proof should use the fixed-spin CMR allowance table, not
% the spin-integrated cluster-count formula. The latter belongs to the
% bond marginal and is recorded separately in Remark~\ref{rem:merge-weights}.

\medskip\noindent
\textbf{Case 1: $q_x=-q_y$.}
Then $\sigma_x\sigma_y=-\tau_x\tau_y$.
For either sign of $J_b$, exactly one replica satisfies $b$.
Conditional on $J_b$ and the frozen spins, the only admissible non-closed state is red, and
\[
  \mathbf{P}_\beta(\eta_b = \textup{red} \mid \mathcal{T}_b)
  = 1-r_b \in (0,1),
  \qquad r_b := e^{-2\beta|J_b|}.
\]
Thus the grey-open probability lies strictly between $0$ and $1$ already at fixed $J_b$.

\medskip\noindent
\textbf{Case 2: $q_x=q_y$.}
Then $\sigma_x\sigma_y=\tau_x\tau_y$.
There is a favorable sign of $J_b$ for which both replicas satisfy $b$, and on that sign $\eta_b$ is blue with conditional probability $1-r_b^2>0$ and closed with probability $r_b^2>0$.
For the opposite sign, neither replica satisfies $b$, so $\eta_b$ is forced closed.
By Proposition~\ref{prop:bayes} with $S=\{b\}$, both signs of $J_b$ have positive conditional probability given $\mathcal{T}_b'$.
The favorable sign gives
\[
  \mathbf{P}_\beta(\eta_b = \textup{blue} \mid \mathcal{T}_b') > 0,
\]
while the unfavorable sign gives
\[
  \mathbf{P}_\beta(\eta_b = \textup{closed} \mid \mathcal{T}_b') > 0.
\]
Therefore the grey-open probability again lies strictly between $0$ and $1$.

This qualitative finite-energy statement gives no uniform lower bound on the grey-open conditional probability, so it does not directly yield stochastic domination by a supercritical Bernoulli process; the Peierls argument of Proposition~\ref{prop:grey-perc} is used instead.

\medskip\noindent
In both cases $0 < \mathbf{P}_\beta(\eta_b \in \{\text{blue}, \text{red}\} \mid \mathcal{T}_b') < 1$, establishing finite energy.
\end{proof}

\begin{remark}[Marginal merge weights]\label{rem:merge-weights}
The finite-energy proof above conditions on the spin lift, so no cluster-counting factors are needed.
After integrating out the spins in finite volume, the CMR bond marginal has fugacities $B_{\mathrm{fug}} = (1-r^2)/r^2$ for blue and $R_{\mathrm{fug}} = (1-r)/r$ for red, per channel, together with the geometric factors in Remark~\ref{rem:weights}.
In the merge case ($\Delta K=-1$), the unnormalized weights relative to the closed state are
\begin{equation}\label{eq:merge-weights}
  \underbrace{1-r^2}_{\text{blue}\;(2^{-2})} \;:\; \underbrace{2r(1-r)}_{\text{red}\;(2^{-1} \times 2\text{ channels})} \;:\; \underbrace{4r^2}_{\text{closed}\;(1)},
\end{equation}
where blue incurs $2^{-2}$ ($\Delta K=\Delta C_b=-1$), red incurs $2^{-1}$ ($\Delta K=-1$, $\Delta C_b=0$), and closed is the reference.
In the cycle case ($\Delta K=\Delta C_b=0$), no geometric factors arise.
\end{remark}

\begin{remark}[Comparison with single-edge disorder averaging]\label{rem:single-bond}
The role of Proposition~\ref{prop:bayes} is analogous to the single-edge disorder average in De~Santis and Gandolfi~\cite[Lemma~3.1]{DSG1999}: in both arguments, summing over the coupling on one edge introduces a partition-function ratio and restores a positive conditional probability. Their lemma concerns a frustrated FK model rather than the frozen-spin CMR conditional used here, so we use it only as a comparison.
\end{remark}

\begin{corollary}[Grey-cluster uniqueness]\label{cor:grey-unique}
The grey subgraph has at most one infinite connected component under $\mathbf{P}_\beta$.
\end{corollary}

\begin{proof}
Proposition~\ref{prop:grey-fe} establishes $\mathbf{P}_\beta(\eta_b \in \{\text{blue, red}\} \mid \mathcal{T}_b') \in (0,1)$ a.s.
By the tower property, the same bounds hold conditional on $\eta_{-b}$ alone (which is coarser than $\mathcal{T}_b'$), giving insertion and deletion tolerance for the grey-bond marginal.
The standard Burton--Keane argument~\cite{BK1989,GKN1992} then applies.
\end{proof}

\begin{corollary}[Subcritical phase]\label{cor:grey-sub}
For $\beta$ sufficiently small, $G_{\mathrm{grey}}$ contains no infinite connected component $\mathbf{P}_\beta$-a.s.
\end{corollary}

\begin{proof}
Fix $M > 0$ and write
\[
  q(M) := \mathscr{P}(|J_e| > M).
\]
Conditional on $(J, \sigma, \tau)$, the bond variables are independent across edges under the Edwards--Sokal coupling, and an edge is grey-open with probability at most
\[
  p_e = 1 - e^{-4\beta |J_e|}.
\]
Let $(U_e)_{e \in E}$ be i.i.d.\ uniform random variables, independent of $(J, \sigma, \tau)$, and realize the grey-open indicator $G_e := \mathbf{1}[\eta_e \in \{\text{blue}, \text{red}\}]$ by the standard monotone coupling.
Then
\[
  G_e \leq \Xi_e := \mathbf{1}\bigl[|J_e| > M \text{ or } U_e \leq 1 - e^{-4\beta M}\bigr]
\]
for every edge $e$: if $|J_e| \leq M$, then $p_e \leq 1 - e^{-4\beta M}$, so grey-open can occur only when $U_e \leq 1 - e^{-4\beta M}$.
Since the couplings are i.i.d.\ by~\textup{(D1)}, the variables $(\Xi_e)_{e \in E}$ are i.i.d.\ Bernoulli with parameter
\[
  p_{\mathrm{dom}}(M,\beta)
  = q(M) + (1-q(M))(1-e^{-4\beta M}).
\]
Choose $M$ large enough that $q(M) < p_c(\mathbb{Z}^d)/2$, then choose $\beta$ small enough that $(1-q(M))(1-e^{-4\beta M}) < p_c(\mathbb{Z}^d)/2$.
For this choice, $p_{\mathrm{dom}}(M,\beta) < p_c(\mathbb{Z}^d)$, so the dominating Bernoulli bond percolation is subcritical.
Therefore $G_{\mathrm{grey}}$ contains no infinite connected component $\mathbf{P}_\beta$-a.s.
\end{proof}

\subsection{Grey percolation transition}

Corollaries~\ref{cor:grey-unique} and~\ref{cor:grey-sub} establish uniqueness and the subcritical phase directly from finite energy.
For the supercritical phase, we adapt the Peierls contour argument of De~Santis and Gandolfi~\cite[Theorem~2.2]{DSG1999}.

\begin{lemma}[Contour partition lemma]\label{lem:contour-partition}
Fix a box $\Lambda$, a Peierls contour $\gamma \subset \Lambda$, and an exterior bond configuration $\eta_{\Lambda \setminus \gamma}$ satisfying~\textup{(C1)}.
Then the contour bonds admit a partition $\gamma = A \sqcup B$ such that opening exactly the bonds of $A$ as red (with $B$ closed) preserves~\textup{(C1)}, and likewise opening exactly the bonds of $B$ as red (with $A$ closed) preserves~\textup{(C1)}.
Consequently every subset of $A$ and every subset of $B$ is red-admissible.
\end{lemma}

\begin{proof}
Regard the occupied exterior graph as a signed graph by assigning sign $+1$ to every blue bond and sign $-1$ to every red bond, and assign sign $-1$ to every contour bond that is proposed to be opened as red.
Condition~\textup{(C1)} says exactly that the occupied exterior configuration is unfrustrated in this auxiliary signed graph: the product of the edge signs around every occupied cycle is~$+1$.
The contour partition is therefore precisely the partition supplied by De~Santis and Gandolfi~\cite[Lemma~2.1]{DSG1999} in two dimensions and its Peierls-surface extension~\cite[Lemma~2.4]{DSG1999} in higher dimensions.
Both $A$ and $B$ can be occupied without producing a negative-sign cycle, which is equivalent here to preserving~\textup{(C1)}.
Deleting occupied contour bonds cannot create a cycle, so every subset of either set is also red-admissible.
\end{proof}

\begin{proposition}[Supercritical grey percolation]\label{prop:grey-perc}
Under~\textup{(D1)--(D2)}, for every $d \geq 2$ and $\beta$ sufficiently large (depending on $d$ and $\mathscr{P}$), $G_{\mathrm{grey}}$ contains exactly one infinite connected component $\mathbf{P}_\beta$-a.s.

For coupling distributions with $|J_e| \geq J_{\min} > 0$ a.s.\ (e.g., the $\pm J$ model), the \emph{existence estimate} is \textbf{quenched}: for every such coupling realization, every infinite-volume CMR bond marginal consistent with the local specification has an infinite grey cluster when
\[
  \beta > \frac{1}{2J_{\min}} \log\!\bigl(2c(d)^2 - 1\bigr),
\]
where $c(2) = 3$ and $c(d) = \exp(64\log d / d)$ for $d \geq 3$.
Uniqueness holds under $\mathbf{P}_\beta$ and hence, by disintegration, under $\zeta_{\beta,J}$ for $\nu$-almost every $J$; no fixed-$J$ uniqueness claim is made for exceptional realizations.
For general coupling laws satisfying~\textup{(D1)--(D2)}, including Gaussian disorder and discrete laws accumulating at zero, the result follows by averaging the quenched Peierls estimate over $J$.
\end{proposition}

\begin{proof}
We adapt the Peierls argument of~\cite[Theorem~2.2]{DSG1999} with one modification: we replace their frustration-based $\mathbb{Z}_2$ partition with a red-bond parity partition.
For $\nu$-almost every $J$, the DLR property of $\rho_{\beta,J}$ together with the local CMR bond kernel makes $\zeta_{\beta,J}$ consistent with the corresponding spin-integrated CMR bond specification.
Work first in a finite-volume CMR bond marginal $\zeta_{\beta,J,\Lambda}$ with the boundary condition induced by the fixed exterior configuration.
Let $\eta^\varnothing$ be the configuration with every contour bond closed, and let $V_\Lambda$ denote the full CMR weight, including its geometric factors.
Then
\[
  \zeta_{\beta,J,\Lambda}(\gamma \equiv 0 \mid \eta_{\Lambda \setminus \gamma})
  = \frac{V_\Lambda(\eta^\varnothing)}
         {\sum_{\eta_\gamma} V_\Lambda(\eta_\gamma)},
\]
where the sum runs over all contour configurations compatible with $\eta_{\Lambda \setminus \gamma}$ and the constraints~(C1)--(C2).
Since restricting the denominator to \emph{any} subset of valid configurations gives a valid upper bound, we restrict to configurations in which every open bond on $\gamma$ is \textbf{red}.
The exterior configuration already satisfies~(C2), and adding red contour bonds cannot create a new all-blue cycle, so~(C2) imposes no further restriction on this family.

\smallskip\noindent
\textbf{Step 1: Parity partition.}
Fix a coupling realization $J$ and a Peierls contour $\gamma$ (the image of a self-avoiding dual circuit in $d = 2$, or a connected plaquette surface in $d \geq 3$; see~\cite[p.~1786, 1791]{DSG1999}).
Let $\Lambda$ be a box containing $\gamma$ and condition on $\eta_{\Lambda \setminus \gamma}$ (all bonds outside $\gamma$).
By Lemma~\ref{lem:contour-partition}, the contour bonds partition as $\gamma = A \sqcup B$ so that every subset of $A$ can be opened as red with $B$ closed without violating~(C1), and likewise every subset of $B$ can be opened as red with $A$ closed.
Denote $|A| = m$; by exchanging $A$ and $B$ if necessary, we may assume $m \leq |\gamma|/2$.

\smallskip\noindent
\textbf{Step 2: Weight ratio.}
For a red-admissible subset $S \subseteq \gamma$, let $\eta^S$ be obtained from $\eta^\varnothing$ by opening exactly the bonds of $S$ as red.
The cluster-weight formula in Remark~\ref{rem:weights} gives
\[
  \frac{V_\Lambda(\eta^S)}{V_\Lambda(\eta^\varnothing)}
  = \prod_{e \in S}\frac{1-r_e}{r_e}
    2^{K(\eta^S)-K(\eta^\varnothing)}.
\]
Opening red bonds does not change $C_b$, and each inserted edge decreases the number of grey components by at most one.  Hence
\[
  K(\eta^S)-K(\eta^\varnothing) \geq -|S|,
\]
and therefore
\begin{equation}\label{eq:inhomogeneous-red-ratio}
  \frac{V_\Lambda(\eta^S)}{V_\Lambda(\eta^\varnothing)}
  \;\geq\; \prod_{e \in S} R_{\mathrm{merge},e},
  \qquad R_{\mathrm{merge},e} := \frac{1-r_e}{2r_e}.
\end{equation}
When $r_e \equiv r$ on the contour, as in the $\pm J$ model, this lower bound is $R_{\mathrm{merge}}^{|S|}$ with
\begin{equation}\label{eq:R-def}
  R_{\mathrm{merge}} := \frac{1-r}{2r}.
\end{equation}

\smallskip\noindent
\textbf{Step 3: Binomial sum.}
Write $\gamma \equiv 0$ for the event that all bonds on $\gamma$ are closed.
Following~\cite[eq.~(2.4)]{DSG1999}, and noting that the two parity-restricted families overlap in the all-closed configuration, the conditional probability of $\gamma \equiv 0$ satisfies
\begin{equation}\label{eq:peierls-denom}
  \zeta_{\beta,J,\Lambda}(\gamma \equiv 0 \mid \eta_{\Lambda \setminus \gamma})
  \;\leq\;
  \Biggl[
    \sum_{k=0}^{m} \binom{m}{k} R_{\mathrm{merge}}^k
    \;+\; \sum_{k=0}^{|\gamma|-m} \binom{|\gamma|-m}{k} R_{\mathrm{merge}}^k
    \;-\; 1
  \Biggr]^{-1}.
\end{equation}
where the first sum runs over subsets of $A$ opened as red (with $B$ closed), the second over subsets of $B$ opened as red (with $A$ closed), and the numerator (all-closed weight) is normalized to~$1$.
By the binomial theorem, each sum evaluates to
\[
  \sum_{k=0}^{n} \binom{n}{k} R_{\mathrm{merge}}^k
  = (1+R_{\mathrm{merge}})^n
  = \Bigl(\frac{1+r}{2r}\Bigr)^{\!n},
\]
so the denominator of~\eqref{eq:peierls-denom} is
\[
  (1+R_{\mathrm{merge}})^m
  + (1+R_{\mathrm{merge}})^{|\gamma|-m} - 1.
\]
Since $m \leq |\gamma|/2$ and $1+R_{\mathrm{merge}} > 1$, the larger term satisfies
\[
  (1+R_{\mathrm{merge}})^{|\gamma|-m}
  \geq (1+R_{\mathrm{merge}})^{|\gamma|/2},
\]
hence
\begin{equation}\label{eq:peierls}
  \zeta_{\beta,J,\Lambda}(\gamma \equiv 0 \mid \eta_{\Lambda \setminus \gamma})
  \;\leq\; \Bigl(\frac{2r}{1+r}\Bigr)^{\!|\gamma|/2}.
\end{equation}
The estimate is uniform in the finite volume and the exterior-induced boundary condition, so it passes to every infinite-volume CMR bond marginal obtained from the corresponding local specification.

\smallskip\noindent
\textbf{Step 4: Contour sum.}
The number of Peierls contours of size $k$ surrounding the origin is at most $k \cdot c(d)^k$~\cite[p.~1788, 1791]{DSG1999}.
Summing over contours:
\[
  \sum_{\gamma \ni O} \zeta_{\beta,J}(\gamma \equiv 0)
  \;\leq\; \sum_{k \geq 1} k\,c(d)^k \Bigl(\frac{2r}{1+r}\Bigr)^{\!k/2}
  \;<\; \infty
\]
provided $c(d)\bigl(\frac{2r}{1+r}\bigr)^{1/2} < 1$, equivalently $r < 1/(2c(d)^2 - 1)$.
By Borel--Cantelli, a.s.\ only finitely many contours surrounding the origin are all-closed.
As in the separation argument of~\cite[p.~1788]{DSG1999}, choose a vertex outside all such contours and adjacent to an outermost one (or use the origin if there is no such contour).  If its grey cluster were finite, the closed edge boundary of that cluster would supply another surrounding contour, a contradiction.  Hence grey percolation occurs.
Under the joint law $\mathbf{P}_\beta$, uniqueness follows from Corollary~\ref{cor:grey-unique}; disintegration then gives uniqueness under $\zeta_{\beta,J}$ for $\nu$-almost every $J$.

For the $\pm J$ model ($|J| = 1$) in $d = 2$: the threshold is $r < 1/17$, i.e., $\beta > \tfrac{1}{2}\ln 17 \approx 1.42$.
This coincides with De~Santis and Gandolfi's threshold $\bar{p}_c^{(2)}(2) \leq 16/17$ for the $q = 2$ FK random-cluster model~\cite[Corollary~2.3]{DSG1999}: their occupied-to-closed ratio $p/[2(1-p)]$ equals our $R_{\mathrm{merge}} = (1-r)/(2r)$ after setting $p=1-r$, giving the same exponential contour factor and hence the same threshold condition.
For $d \geq 3$, the contour-counting bound above gives the stated threshold with $c(d)=\exp(64\log d/d)$.

More generally, if $|J_e|\geq J_{\min}>0$, set $r_*:=e^{-2\beta J_{\min}}$.  Since $r_e\leq r_*$ and $r\mapsto 2r/(1+r)$ is increasing, the edge-dependent estimate below is bounded by the homogeneous expression with $r_*$, yielding the threshold in the proposition.

\smallskip\noindent
\textit{Annealed extension.}
For a general coupling law satisfying~\textup{(D1)--(D2)}, the values $r_e = e^{-2\beta|J_e|}$ may vary across edges and approach~$1$.
The quenched bound~\eqref{eq:peierls} generalizes to
\[
  \zeta_{\beta,J}(\gamma \equiv 0)
  \;\leq\;
  \prod_{e \in \gamma} \Bigl(\frac{2r_e}{1+r_e}\Bigr)^{\!1/2},
\]
Indeed,~\eqref{eq:inhomogeneous-red-ratio} makes the two parity-restricted contributions at least
\[
  X_A:=\prod_{e\in A}(1+R_{\mathrm{merge},e}),
  \qquad
  X_B:=\prod_{e\in B}(1+R_{\mathrm{merge},e}).
\]
Their sum minus the duplicated all-closed term is at least $\max\{X_A,X_B\}\geq\sqrt{X_AX_B}$, which gives the displayed product bound.
Taking the expectation over $J$ and using independence~(D1):
\[
  \mathbf{P}_\beta(\gamma \equiv 0)
  \;\leq\; \alpha(\beta)^{|\gamma|},
  \qquad
  \alpha(\beta) := \mathbb{E}\Bigl[\Bigl(\frac{2r_e}{1+r_e}\Bigr)^{\!1/2}\Bigr].
\]
As $\beta \to \infty$, $r_e \to 0$ pointwise for every $J_e \neq 0$, and the integrand is bounded by~$1$.
By dominated convergence, $\alpha(\beta) \to \mathscr{P}(\{0\}) = 0$ (using~(D2)).
Hence $\alpha(\beta) < 1/c(d)$ for $\beta$ large enough, and the contour sum converges.
\end{proof}

\begin{remark}[Possible optimization]
The Peierls bound above is deliberately coarse and appears to lose in two different places. First, on contour edges whose endpoints are already connected through the exterior configuration, the parity constraint from~(C1) restricts the admissible red subsets to one side of a $\mathbb{Z}_2$ balance partition. Second, we price every red contour bond by the worst-case merge ratio $R_{\mathrm{merge}} = \frac{1-r}{2r}$, even though a red bond that only closes a cycle and does not merge grey components should carry the larger ratio $R_{\mathrm{cycle}} = \frac{1-r}{r}$. These two losses should be partially separable: the balance constraint acts on the cycle-type part of the contour, while the geometric factor is only paid by bonds that genuinely merge grey components. A sharper argument would encode each exterior grey component together with the red parity of its contour attachment points as a finite signed quotient, then count balanced red subgraphs with edge weights $R_{\mathrm{cycle}}$ on cycle-type edges and $R_{\mathrm{merge}}$ on genuinely merging edges. We do not pursue this refinement here.
\end{remark}

\begin{proof}[Proof of Theorem~\ref{thm:main}\textup{(a)}]
For $\beta < \beta_{\mathrm{sub}}$, Corollary~\ref{cor:grey-sub} gives no infinite grey cluster.
For $\beta > \beta_{\mathrm{grey}}$, Proposition~\ref{prop:grey-perc} gives at least one infinite grey cluster, and Corollary~\ref{cor:grey-unique} gives at most one.
\end{proof}

\section{A multicolour Burton--Keane argument for blue clusters}\label{sec:bk}

We first isolate the blue-cluster uniqueness mechanism as a statement about labelled random subgraphs.
The abstract argument uses only label-class coalescence and the mass-transport bound on the number of ends; the CMR-specific input is then confined to a finite-box resampling lemma that verifies coalescence for the overlap labels $q=\pm1$.

\subsection{Labelled subgraphs and label-class coalescence}\label{sec:label-coalescence}

Let $\mathcal{A}$ be a finite label set.
A \emph{labelled subgraph} of $\mathbb{Z}^d$ is a pair $(\ell,\omega)$ with vertex labels $\ell\in\mathcal{A}^V$ and edge indicators $\omega\in\{0,1\}^E$.
It is \emph{compatible} if
\[
  \omega_{\{x,y\}}=1 \quad\Longrightarrow\quad \ell_x=\ell_y.
\]
For $a\in\mathcal{A}$, let $G_a$ be the subgraph consisting of the vertices labelled $a$ and the open edges between them.

For a finite box $\Lambda\subset\mathbb{Z}^d$, define its outside vertex boundary by
\[
  \partial^+\Lambda
  :=\bigl\{y\in\Lambda^c: \{x,y\}\in E
        \text{ for some }x\in\Lambda\bigr\}.
\]
Let $G_a^{\Lambda^c}$ be the subgraph of $G_a$ induced by $\Lambda^c$, and let
\[
  \mathscr{G}_{\Lambda^c}
  :=\sigma\bigl(\ell_x:x\in\Lambda^c;\,
                 \omega_e:e\subset\Lambda^c\bigr)
\]
be the labelled exterior sigma-field.
Write $\mathsf{Coal}_a(\Lambda)$ for the event that all infinite components of $G_a^{\Lambda^c}$ meeting $\partial^+\Lambda$ belong to a single connected component of the full graph $G_a$.

\begin{definition}[Label-class coalescence]\label{def:label-coalescence}
The law of a compatible labelled subgraph has \emph{label-class coalescence} if, for every $a\in\mathcal{A}$ and every finite box $\Lambda$,
\[
  \mathbf{P}\bigl(\mathsf{Coal}_a(\Lambda)
       \mid\mathscr{G}_{\Lambda^c}\bigr)>0
  \qquad \mathbf{P}\text{-a.s.}
\]
Thus the property asks only for a positive-probability collective merger within one fixed label class.
It neither permits arbitrary single-edge insertions nor requires mergers between different labels.
\end{definition}

This is a label-restricted analogue of the merge tolerance of Halberstam and Hutchcroft~\cite[Lemma~3.4]{HH2023}, weaker in what it can merge: their condition can merge arbitrary augmented connectivity classes meeting a finite connected subgraph, whereas Definition~\ref{def:label-coalescence} concerns only components carrying one prescribed label.

\subsection{Abstract multicolour uniqueness proposition}\label{sec:abstract-bk}

\begin{proposition}[Label-class Burton--Keane]\label{prop:label-bk}
Let $(\ell,\omega)$ be an ergodic translation-invariant compatible labelled random subgraph of $\mathbb{Z}^d$ with finite label set $\mathcal{A}$.
If its law has label-class coalescence, then for every $a\in\mathcal{A}$ the graph $G_a$ has at most one infinite connected component almost surely.
Consequently the full open subgraph has at most $|\mathcal{A}|$ infinite connected components almost surely.
\end{proposition}

\begin{proof}
Fix $a\in\mathcal{A}$, and let $N_a$ be the number of infinite components of $G_a$.
Ergodicity makes $N_a$ an almost-sure constant in $\{0,1,2,\ldots,\infty\}$.
For a finite box $\Lambda$, let $A_{a,m}(\Lambda)$ be the event that $G_a^{\Lambda^c}$ has at least $m$ distinct infinite components meeting $\partial^+\Lambda$.
Let $D_a(\Lambda)$ be the event that $G_a^{\Lambda^c}$ has at least one infinite component and that every infinite component of $G_a^{\Lambda^c}$ meets $\partial^+\Lambda$.
Both events belong to $\mathscr{G}_{\Lambda^c}$.

Suppose first that $N_a=k$ for some $2\leq k<\infty$.
Almost surely some finite box meets all $k$ infinite components, so for some deterministic $\Lambda$ this occurs with positive probability.
After deleting that box, every infinite exterior component meets $\partial^+\Lambda$; hence $\mathbf{P}(D_a(\Lambda))>0$.
Label-class coalescence and exterior measurability give
\[
  \mathbf{P}\bigl(D_a(\Lambda)\cap
                   \mathsf{Coal}_a(\Lambda)\bigr)>0.
\]
On this event all infinite exterior components are joined into one component of $G_a$.
Because $\mathbb{Z}^d$ is locally finite, every infinite component left after changing finitely many vertices and incident edges contains an infinite component of the exterior graph.
Thus $N_a=1$ on this positive-probability event, contradicting the almost-sure constancy $N_a=k$.

Suppose instead that $N_a=\infty$.
Almost surely some finite box meets at least three infinite components, and therefore $\mathbf{P}(A_{a,3}(\Lambda))>0$ for some deterministic $\Lambda$.
Again by exterior measurability and label-class coalescence,
\[
  \mathbf{P}\bigl(A_{a,3}(\Lambda)\cap
                   \mathsf{Coal}_a(\Lambda)\bigr)>0.
\]
On this event a single component of $G_a$ contains three distinct infinite components of $G_a^{\Lambda^c}$.
Removing $\Lambda$ leaves these three components distinct and infinite, so the full component has at least three ends.
Viewing vertices with labels other than $a$ as isolated, $G_a$ is a translation-invariant random subgraph of $\mathbb{Z}^d$.
This contradicts~\cite[Proposition~3.3]{HH2023}, which states that every connected component of such a subgraph has at most two ends almost surely.

Hence $N_a\in\{0,1\}$.
Applying this conclusion to every $a\in\mathcal{A}$ proves the final assertion because compatibility prevents an open component from containing two labels.
\end{proof}

\subsection{CMR verification}\label{sec:cmr-coalescence}

We now return to the CMR joint law under the ergodic reduction following Theorem~\ref{thm:main}.
Take $\mathcal{A}=\{+1,-1\}$, set $\ell_x=q_x=\sigma_x\tau_x$, and let $\omega_e$ indicate that $e$ is blue.
The allowance rules in Definition~\ref{def:CMR} make this a compatible labelled subgraph.

\begin{remark}[Clock-state encoding]
Encoding $(+,+),(+,-),(-,-),(-,+)$ by $0,1,2,3\in\mathbb{Z}_4$, respectively, gives $q_x=(-1)^{z_x}$.
Thus the overlap is the parity projection of a four-state clock variable, but the argument uses only the resulting two-valued label map and no group structure.
\end{remark}

\begin{lemma}[CMR box filling]\label{lem:resample}
Fix $s\in\{\pm1\}$.
Let $\Lambda\subset\mathbb{Z}^d$ be a finite box, and let
\[
  E_\Lambda:=\{e\in E:e\cap\Lambda\neq\emptyset\}
\]
be the set of edges touching $\Lambda$, including its boundary edges.
Let $\mathscr{F}_{\Lambda^c}$ be the sigma-field generated by the couplings and bond variables on edges with both endpoints in $\Lambda^c$ and by the spins on $\Lambda^c$.
Let $B_s(\Lambda)$ be the event that:
\begin{enumerate}
\item[\textup{(i)}] every vertex of $\Lambda$ has overlap $q_i=s$;
\item[\textup{(ii)}] every edge with both endpoints in $\Lambda$ is blue;
\item[\textup{(iii)}] every boundary edge whose outside endpoint has overlap $s$ is blue;
\item[\textup{(iv)}] every boundary edge whose outside endpoint has overlap $-s$ is non-blue.
\end{enumerate}
Then
\[
  \mathbf{P}_\beta\bigl(B_s(\Lambda)
       \mid\mathscr{F}_{\Lambda^c}\bigr)>0
  \qquad \mathbf{P}_\beta\text{-a.s.}
\]
On $B_s(\Lambda)$, all exterior blue components in the overlap class $q=s$ that meet $\partial^+\Lambda$ are joined through $\Lambda$.
\end{lemma}

\begin{proof}
Fix the outside configuration in $\mathscr{F}_{\Lambda^c}$.
We build the desired event in three steps.

\medskip\noindent
\textbf{Step 1: prescribe the spins in $\Lambda$.}
Let
\[
  (\sigma_i,\tau_i)=
  \begin{cases}
    (+1,+1),&s=+1,\\
    (+1,-1),&s=-1,
  \end{cases}
  \qquad i\in\Lambda.
\]
Conditional on all couplings, the finite-box two-replica DLR kernel assigns strictly positive weight to every spin pattern.
Averaging over the unfrozen couplings on $E_\Lambda$ preserves strict positivity, so this event has positive conditional probability given $\mathscr{F}_{\Lambda^c}$.

\medskip\noindent
\textbf{Step 2: prescribe the couplings on $E_\Lambda$.}
On the event from Step~1, require
\begin{itemize}
\item $J_e>0$ for every edge $e$ with both endpoints in $\Lambda$;
\item $J_e\sigma_j>0$ for every boundary edge $e=\{i,j\}$ with $i\in\Lambda$ and $j\in\Lambda^c$.
\end{itemize}
This sign pattern has positive $\mathscr{P}^{\otimes E_\Lambda}$-measure by~\textup{(D2)}.
Once the spins and exterior configuration have been frozen, the conditioning field is exactly $\mathcal{T}'_{E_\Lambda}$ from Proposition~\ref{prop:bayes}.
That proposition therefore gives the prescribed sign pattern positive conditional probability.

\medskip\noindent
\textbf{Step 3: realize the bond states on $E_\Lambda$.}
Given the preceding spin and coupling choices, every interior edge is satisfied by both replicas.
A boundary edge is satisfied by both replicas exactly when its outside endpoint has overlap $s$; when that endpoint has overlap $-s$, exactly one replica satisfies it.
Thus every edge required to be blue in~\textup{(ii)}--\textup{(iii)} is blue with probability
$1-e^{-4\beta|J_e|}>0$, while the edges in~\textup{(iv)} are automatically non-blue.
The bond variables are conditionally independent given $(J,\sigma,\tau)$, and $E_\Lambda$ is finite, so the conjunction has positive conditional probability.

Successive conditioning through the three steps proves the displayed claim.
The joining statement follows from~\textup{(i)}--\textup{(iii)}.
\end{proof}

\begin{proposition}[CMR label-class coalescence]\label{prop:cmr-label-coalescence}
The labelled blue-subgraph factor $(q,\omega)$ of an ergodic component of an admissible CMR joint law is ergodic, translation-invariant, compatible, and satisfies label-class coalescence.
\end{proposition}

\begin{proof}
Ergodicity and translation invariance pass to the shift-equivariant factor $(q,\omega)$, and compatibility was noted above.
Its exterior field is
\[
  \mathscr{G}_{\Lambda^c}
  =\sigma\bigl(q_x:x\in\Lambda^c;\,
               \mathbf{1}_{\{\eta_e=\mathrm{blue}\}}:e\subset\Lambda^c\bigr),
\]
which is contained in the richer joint field $\mathscr{F}_{\Lambda^c}$.
Moreover $B_s(\Lambda)\subseteq\mathsf{Coal}_s(\Lambda)$.
The tower property and Lemma~\ref{lem:resample} therefore give
\begin{align*}
  \mathbf{P}_\beta\bigl(\mathsf{Coal}_s(\Lambda)
      \mid\mathscr{G}_{\Lambda^c}\bigr)
  &\geq \mathbf{P}_\beta\bigl(B_s(\Lambda)
      \mid\mathscr{G}_{\Lambda^c}\bigr)\\
  &=\mathbf{E}_\beta\!\left[
      \mathbf{P}_\beta\bigl(B_s(\Lambda)
          \mid\mathscr{F}_{\Lambda^c}\bigr)
      \mathrel{\big|}\mathscr{G}_{\Lambda^c}\right]>0
\end{align*}
almost surely.
This is Definition~\ref{def:label-coalescence} for $s=\pm1$.
\end{proof}

\begin{figure}[H]
\centering
\begin{tikzpicture}[
    scale=0.82,
    bluebond/.style={draw=blue!70!black, line width=1.15pt},
    otherparity/.style={draw=blue!70!black, densely dotted, line width=1.05pt},
    resamplegrid/.style={draw=blue!45!black, line width=0.45pt},
    box/.style={draw=black, thick},
    vertex/.style={circle, fill=blue!70!black, inner sep=1.25pt},
    othervertex/.style={circle, draw=blue!70!black, fill=white, line width=0.8pt, inner sep=1.0pt},
    every node/.style={font=\small}
  ]

  \begin{scope}
    \node at (1.5,2.95) {three exterior $q=s$ components};
    \draw[box] (0,0) rectangle (3,2.2);
    \node at (1.5,1.1) {$\Lambda$};

    \draw[bluebond] (-1.25,2.75) -- (-0.4,2.45) -- (0.55,2.2);
    \draw[bluebond] (-1.25,-0.35) -- (-0.35,-0.15) -- (0.45,0);
    \draw[bluebond] (4.2,1.15) -- (3.5,1.2) -- (3,1.25);
    \draw[otherparity] (4.2,2.55) -- (3.45,2.35) -- (3,2.05);
    \draw[otherparity] (-1.75,1.58) -- (-0.82,1.48) -- (0,1.32);

    \node[vertex] at (0.55,2.2) {};
    \node[vertex] at (0.45,0) {};
    \node[vertex] at (3,1.25) {};
    \node[othervertex] at (3,2.05) {};
    \node[othervertex] at (0,1.32) {};

    \node[blue!70!black] at (-1.15,2.15) {$q=s$};
    \node[blue!70!black, fill=white, inner sep=0.7pt] at (-1.1,1.52) {$q=-s$};
    \node[blue!70!black] at (-1.1,0.35) {$q=s$};
    \node[blue!70!black] at (4.1,0.75) {$q=s$};
    \node[blue!70!black] at (4.1,2.2) {$q=-s$};

    \node[align=center, font=\scriptsize, text width=4.0cm] at (1.5,-0.72)
      {same-label components meet $\partial^+\Lambda$};
  \end{scope}

  \draw[-{Stealth[length=6pt]}, thick] (4.85,1.1) -- (6.35,1.1)
    node[midway, below, fill=white, inner sep=1pt] {$B_s(\Lambda)$};

  \begin{scope}[xshift=8.0cm]
    \node at (1.5,2.95) {box resampled};
    \draw[box] (0,0) rectangle (3,2.2);
    \draw[resamplegrid, step=0.55] (0.35,0.35) grid (2.65,1.85);
    \node[fill=white, inner sep=1pt] at (1.5,1.1) {$q=s$};
    \coordinate (merge) at (1.5,1.1);

    \draw[bluebond] (-1.25,2.75) -- (-0.4,2.45) -- (0.55,2.2);
    \draw[bluebond] (-1.25,-0.35) -- (-0.35,-0.15) -- (0.45,0);
    \draw[bluebond] (4.2,1.15) -- (3.5,1.2) -- (3,1.25);
    \draw[otherparity] (4.2,2.55) -- (3.45,2.35) -- (3,2.05);
    \draw[otherparity] (-1.75,1.58) -- (-0.82,1.48) -- (0,1.32);

    \draw[bluebond] (0.55,2.2) -- (0.55,1.85) -- (1.5,1.85) -- (merge);
    \draw[bluebond] (0.45,0) -- (0.45,0.35) -- (1.5,0.35) -- (merge);
    \draw[bluebond] (3,1.25) -- (2.65,1.25) -- (merge);

    \node[vertex] at (0.55,2.2) {};
    \node[vertex] at (0.45,0) {};
    \node[vertex] at (3,1.25) {};
    \node[othervertex] at (3,2.05) {};
    \node[othervertex] at (0,1.32) {};
    \node[vertex] at (merge) {};

    \node[align=center, font=\scriptsize, text width=4.0cm] at (1.5,-0.72)
      {all incident $q=s$ arms are merged};
  \end{scope}
\end{tikzpicture}
\caption{CMR box filling verifies label-class coalescence. Given the exterior configuration, the event $B_s(\Lambda)$ has positive conditional probability under the richer joint exterior field and joins every incident exterior blue component carrying the label $q=s$. The dotted blue arms carry the opposite label and are left unmerged.}
\label{fig:box-resampling}
\end{figure}
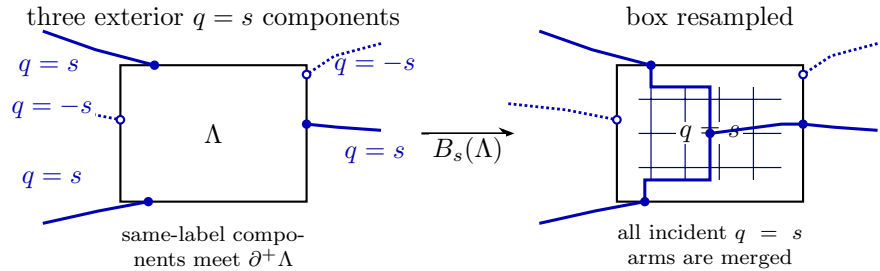

\begin{remark}
The box-filling construction is intentionally stronger than Definition~\ref{def:label-coalescence}: it joins every incident exterior component with label $s$, not only the infinite ones.
It is also not optimized, since all spins in $\Lambda$ and all coupling signs on $E_\Lambda$ are prescribed.
Only strict conditional positivity is needed.
\end{remark}

\begin{proof}[Proof of Theorem~\ref{thm:main}\textup{(b)}]
Under the ergodic reduction, Proposition~\ref{prop:cmr-label-coalescence} verifies the hypotheses of Proposition~\ref{prop:label-bk} for the two labels $q=+1$ and $q=-1$.
Thus there is at most one infinite blue cluster in each overlap class and at most two in total.

If two distinct infinite blue clusters exist, Corollary~\ref{cor:grey-unique} implies that they lie in the same infinite grey cluster.
Within that grey cluster, the two path-parity classes coincide under the spin lift with the overlap classes $q=+1$ and $q=-1$.
The two infinite blue clusters must therefore lie in opposite path-parity classes; equivalently, every grey path between them crosses an odd number of red bonds.

Let $A$ be the event that there are at most two infinite blue clusters and that, whenever two exist, they lie in the same grey cluster and every grey path between them crosses an odd number of red bonds.
This event is measurable from the bond configuration alone, since grey connectivity and red-bond path parity are bond-measurable by Remark~\ref{rem:parity-marginal}.
Hence the conclusion also holds under the bond marginal of $\mathbf{P}_\beta$.
The general translation-invariant case follows by integrating over the ergodic components.
\end{proof}

\subsection{Further consequences}\label{sec:bk-consequences}

Fix $\beta > 0$ and let $\mathbf{P}_\beta$ be as in Theorem~\ref{thm:main}.

\begin{corollary}[Replica-symmetric ergodic case]\label{cor:replica-symmetric}
Assume in addition that $\mathbf{P}_\beta$ is ergodic and invariant under the global replica flip
\[
  (J,\sigma,\tau,\eta) \longmapsto (J,\sigma,-\tau,\eta).
\]
Then:
\begin{enumerate}
\item[\textup{(i)}] $G_{\mathrm{blue}}$ contains either $0$ or $2$ infinite blue clusters $\mathbf{P}_\beta$-a.s.
\item[\textup{(ii)}] If $\rho_+^\infty$ and $\rho_-^\infty$ denote the spatial densities of the infinite blue clusters in the overlap classes $q=+1$ and $q=-1$, respectively, then
\[
  \rho_+^\infty = \rho_-^\infty
  \qquad
  \mathbf{P}_\beta\text{-a.s.}
\]
\end{enumerate}
\end{corollary}

\begin{proof}
Let $N_+$ and $N_-$ denote the numbers of infinite blue clusters in the overlap classes $q=+1$ and $q=-1$, respectively.
By Theorem~\ref{thm:main}\textup{(b)}, $N_+,N_- \in \{0,1\}$ $\mathbf{P}_\beta$-a.s.
Since both are translation-invariant and $\mathbf{P}_\beta$ is ergodic, they are almost surely constant.
The replica-flip map swaps the $q=+1$ and $q=-1$ classes while preserving the bond configuration, so it swaps $N_+$ and $N_-$.
By invariance, $N_+$ and $N_-$ have the same law, hence the same almost-sure constant value.
Therefore $N_+ = N_-$ and the total number of infinite blue clusters is either $0$ or $2$, proving~\textup{(i)}.

For $s \in \{\pm1\}$, let $I_s(x)$ be the indicator that $x$ belongs to the infinite blue cluster in the overlap class $q=s$ (with $I_s \equiv 0$ if no such cluster exists).
By Theorem~\ref{thm:main}\textup{(b)}, there is at most one such cluster, so $I_s(x)$ is well-defined.
Set
\[
  \theta_s := \mathbf{P}_\beta\bigl(0 \text{ belongs to the infinite blue cluster in the class } q=s\bigr)
  = \mathbb{E}_{\mathbf{P}_\beta}[I_s(0)].
\]
By the pointwise ergodic theorem, for the boxes $\Lambda_n := [-n,n]^d \cap \mathbb{Z}^d$,
\[
  \rho_s^\infty
  := \lim_{n\to\infty} \frac{1}{|\Lambda_n|}\sum_{x\in\Lambda_n} I_s(x)
  = \theta_s
  \qquad
  \mathbf{P}_\beta\text{-a.s.}
\]
The replica-flip map swaps $I_+$ and $I_-$, so by invariance $\theta_+ = \theta_-$.
Hence $\rho_+^\infty = \rho_-^\infty$ $\mathbf{P}_\beta$-a.s., proving~\textup{(ii)}.
\end{proof}

\begin{remark}
The replica-flip invariance alone already implies $\theta_+ = \theta_-$, where $\theta_s := \mathbf{P}_\beta(0 \in \mathcal{C}_s^\infty)$ and $\mathcal{C}_s^\infty$ denotes the infinite blue cluster in the class $q=s$ when it exists. Ergodicity is used only to identify these origin probabilities with the almost-sure spatial densities $\rho_s^\infty$ via the pointwise ergodic theorem, and likewise to rule out the possibility of exactly one infinite blue cluster by forcing $N_+$ and $N_-$ to be almost surely constant.
\end{remark}

Since the configuration space is a standard Borel product space and the joint measure $\mathbf{P}_\beta$ has marginal $\nu$ on the couplings $J$, it admits a disintegration
\[
  \mathbf{P}_\beta = \int \nu(\mathrm{d}J)\, \mathbf{P}_{\beta,J},
\]
where $\mathbf{P}_{\beta,J}$ is the conditional measure on $(\sigma, \tau, \text{bonds})$ given $J$.
The following quenched consequence of Theorem~\ref{thm:main} is then immediate by disintegration, equivalently by the tower property for conditional expectation.

\begin{corollary}[Quenched cluster bound]\label{cor:quenched}
For $\nu$-almost every coupling realization $J$, the selected quenched bond marginal $\zeta_{\beta,J}$ supports at most two infinite blue-bond clusters almost surely.
If two distinct infinite blue clusters exist, then they lie in the same grey cluster and every grey path between them crosses an odd number of red bonds.
\end{corollary}

\begin{proof}
Let $A$ be the bond-measurable event appearing in Theorem~\ref{thm:main}\textup{(b)}.
It is measurable with respect to the bond configuration alone, since grey connectivity and red-bond path parity are determined in the bond marginal (Remark~\ref{rem:parity-marginal}).
Theorem~\ref{thm:main} gives $\mathbf{P}_\beta(A) = 1$.
Disintegrating over the disorder,
\[
  1 = \mathbf{P}_\beta(A) = \int \mathbf{P}_{\beta,J}(A)\, \nu(\mathrm{d}J),
\]
so $\mathbf{P}_{\beta,J}(A) = 1$ for $\nu$-a.e.\ $J$, and the same holds under its bond marginal $\zeta_{\beta,J}$.
  In particular, the bound applies to the conditional quenched states obtained by disintegrating any admissible joint measure.
\end{proof}

\section{Discussion}\label{sec:discussion}

Machta, Newman, and Stein~\cite{MNS2007} argue on the basis of mean-field calculations and numerics that the spin glass transition is accompanied by the appearance of exactly two giant blue clusters of unequal density.
M\"unster and Weigel~\cite{MW2023,MW2025} observe two dominant Chayes--Machta--Redner--J\"org (CMRJ) clusters in two- and three-dimensional short-range Ising spin-glass models, with density difference approaching the overlap $|q|$ at low temperature.
For two replicas, their CMRJ clusters coincide with the blue clusters considered here.
Theorem~\ref{thm:main} provides a rigorous upper bound consistent with these observations: the number of infinite blue clusters cannot exceed two, and if two exist then they lie in the common infinite grey cluster and occupy opposite path-parity classes there.

Part~(a) establishes a grey-bond percolation transition, but this does not by itself imply the existence of infinite blue clusters or a spin glass phase transition---the infinite grey cluster could in principle be sustained entirely by red bonds.
Theorem~\ref{thm:main} does not prove the existence of any infinite blue cluster, the existence of a spin glass transition, or any property of the critical point (value of $T_c$, sharpness, critical exponents).

Turning the percolation picture into an identity between the overlap order parameter and blue-cluster densities requires two additional inputs that are not provided here.
First, the overlap order parameter $q = D_+ - D_-$ (where $D_\pm$ is the density of sites with $q_i = \pm 1$) vanishes almost surely under any translation-invariant ergodic joint Gibbs measure that is invariant under $(\sigma, \tau) \mapsto (\sigma, -\tau)$: translation ergodicity forces $D_\pm$ to be a.s.\ constant, the $\mathbb{Z}_2$ symmetry swaps them, so $D_+ = D_- = \tfrac{1}{2}$.
Thus a nontrivial density interpretation of $q$ necessarily requires symmetry-broken (extremal) measures, which our framework does not directly construct.
Second, even in such measures, one would still need to show that finite blue clusters inside the infinite grey cluster do not contribute to the overlap density.
This is exactly the point flagged by Machta, Newman, and Stein~\cite[p.~120]{MNS2007} as ``not immediate'' in short-range models.
Unlike the ferromagnetic FK representation, finite CMR blue clusters do not receive independent random spin labels, so the usual FK cancellation mechanism is unavailable.

Theorem~\ref{thm:main} applies to the periodic-torus joint limits defined in~\eqref{eq:torus-joint-law}, and Corollary~\ref{cor:quenched} covers the conditional quenched laws obtained by disintegrating such a limit.
It does not claim to cover every translation-invariant Gibbs state.
The abstract Proposition~\ref{prop:label-bk} is not specific to CMR: it applies to any ergodic translation-invariant compatible labelled subgraph with label-class coalescence.
The CMR verification extends to other translation-invariant joint CMR laws that satisfy the finite-coupling equivalence of Proposition~\ref{prop:bayes}, finite-box non-nullness of the two-replica spin specification, and the same local CMR bond kernel.

The theorem isolates the part of the CMR picture that is purely structural.
If infinite blue clusters occur in a translation-invariant CMR state, their multiplicity is forced by the overlap parity: at most one can live in each of the two parity classes.
In this sense the two-cluster scenario is not merely a mean-field or numerical artifact: when two infinite blue clusters occur, opposite overlap parity is the only possible multiplicity pattern compatible with the CMR constraints.

The remaining obstruction is not multiplicity, but density.
To identify the spin overlap with a difference of blue-cluster densities, one still has to show that finite blue clusters do not contribute to the overlap density, or else work in a setting where this contribution can be controlled.
This is the short-range difficulty emphasized by Machta, Newman, and Stein, and it remains the natural next problem after the upper-bound result proved here.

\section*{Acknowledgments}
The author thanks Charles M. Newman, Alberto Gandolfi, Daniel L. Stein, and Thomas Hutchcroft for helpful comments on this note.

\bibliographystyle{unsrt}
\bibliography{refs}

\end{document}